\documentclass[12pt,intlimits]{article}

\usepackage{amssymb}
\usepackage{amsthm}
\usepackage{amsmath}

\usepackage{dsfont}

\newcommand{\N}{\mathds{N}}
\newcommand{\Z}{\mathds{Z}}
\newcommand{\Q}{\mathds{Q}}
\newcommand{\R}{\mathds{R}}

\newcommand{\1}{\mathds{1}}

\newcommand{\loc}{\mathrm{loc}}
\newcommand{\spt}{\mathrm{spt}}

\providecommand{\abs}[1]{\left\lvert#1\right\rvert}
\providecommand{\norm}[1]{\left\lVert#1\right\rVert}
\providecommand{\set}[1]{\left\{ #1\right\}}
\providecommand{\form}{\tau}
\providecommand{\scpr}[2]{\left( #1 \,\middle|\, #2 \right)}
\renewcommand{\sp}{\scpr}

\newcommand\rlim{
\mathchoice{\vcenter{\hbox{${\scriptstyle{+}}$}}}
{\vcenter{\hbox{$\scriptstyle{+}$}}}
{\vcenter{\hbox{$\scriptscriptstyle{+}$}}}
{\vcenter{\hbox{$\scriptscriptstyle{+}$}}}}

\makeatletter

\theoremstyle{plain} 
\newtheorem{theorem}{Theorem}[section]

\newtheorem{lemma}[theorem]{Lemma}

\theoremstyle{definition}
\newtheorem{example}[theorem]{Example}
\newtheorem{definition}{Definition}

\newtheorem{remark}[theorem]{Remark}

\makeatother

\usepackage{enumitem}

\setenumerate{fullwidth}
\setenumerate[1]{itemindent=2\parindent, labelwidth=1\parindent, nolistsep} 
\setenumerate[2]{nolistsep} 

\begin{document}

\title{Gordon type Theorem for measure perturbation}
\author{Christian Seifert}

\maketitle

\begin{abstract}
	Generalizing the concept of Gordon potentials to measures we prove a version of Gordon's theorem for measures as potentials and show absence of eigenvalues for these one-dimensional Schr\"odinger operators.

	\vspace{8pt}
	
	\noindent
	MSC 2010: 34L05, 34L40, 81Q10 

	\vspace{2pt}
	
	\noindent
	Keywords: Schr\"odinger operators, eigenvalue problem, quasiperiodic measure potentials
\end{abstract}

\section{Introduction}
\label{sec:introduction}

According to \cite{DamanikStolz2000}, the one-dimensional Schr\"odinger 
operator $H=-\Delta+V$ has no eigenvalues if the potential $V\in 
L_{1,\loc}(\R)$ can be approximated by periodic potentials (in a suitable 
sense). The aim of this paper is to generalize this result to measures $\mu$ 
instead of potential functions $V$, i.e., to more singular potentials.

Although all statements remain valid for complex measures we only focus on real (but signed) measures $\mu$, 
since we are interested in self-adjoint operators.

In the remaining part of this section we explain the situation and define the 
operator in question. We also describe the class of measures we are 
concerned with.
Section \ref{sec:solution_estimates} provides all the tools we need to prove 
the main theorem: $H=-\Delta+\mu$ has no eigenvalues for suitable $\mu$.
In section \ref{sec:examples} we show some examples for Schr\"odinger operators with measures as potentials.

We consider a Schr\"odinger operator of the form
\[ H = -\Delta + \mu\]
on $L_2(\R)$. Here, $\mu=\mu_+ - \mu_-$ is a signed Borel measure on $\R$ with 
locally finite total variation $\abs{\mu}$.

We define $H$ via form methods. To this end, we need to establish form 
boundedness of $\mu_-$. Therefore, we restrict the class of measures we want to
consider.

\begin{definition}
	A signed Borel measure $\mu$ on $\R$ is called \emph{uniformly locally bounded}, if
	\[\norm{\mu}_{\loc} := \sup_{x\in \R} \abs{\mu}([x,x+1]) < \infty.\]
	We call $\mu$ a \emph{Gordon measure} if $\mu$ is uniformly locally bounded
	and if there exists a sequence $(\mu^{m})_{m\in\N}$ of uniformly locally bounded periodic 
	Borel measures with period sequence $(p_m)$ such that $p_m \to \infty$ and 
	for all $C\in \R$ we have
	\[\lim_{m\to\infty} e^{Cp_m}\abs{\mu-\mu^{m}}([-p_m,2p_m]) = 0,\]
	i.e., $(\mu^{m})$ approximates $\mu$ on increasing intervals.
	Here, a Borel measure is $p$-periodic, if $\mu = \mu(\cdot + p)$.
\end{definition}

Clearly, every generalized Gordon potential $V\in L_{1,\loc}$ as defined in 
\cite{DamanikStolz2000} induces a Gordon measure $\mu=V\,\lambda$, where 
$\lambda$ is the Lebegue measure on $\R$. Therefore, also every Gordon 
potential (see the original work \cite{Gordon1976}) induces a Gordon measure.

\begin{lemma}
	Let $\mu$ be a uniformly locally bounded measure. Then $\abs{\mu}$ is $-\Delta$-form bounded, and for all $0<c<1$ there is $\gamma\geq 0$ such that
	\[\int_{\R} \abs{u}^2\, d\abs{\mu} \leq c \norm{u'}_2^2 + \gamma \norm{u}_2^2 \quad(u\in W_{2}^1(\R)).\]
\end{lemma}

\begin{proof}
	For $\delta\in(0,1)$ and $n\in \Z$ we have
	\[\norm{u}_{\infty,[n\delta,(n+1)\delta]}^2 \leq 4 \delta \norm{u'}_{L_2(n\delta,(n+1)\delta)}^2 + \frac{4}{\delta} \norm{u}_{L_2(n\delta,(n+1)\delta)}^2\]
	by Sobolev's inequality.
	
	Now, we estimate
	\begin{align*}
		\int_{\R} \abs{u}^2\, d\abs{\mu} & = \sum_{n\in \Z} \int_{n\delta}^{(n+1)\delta} \abs{u}^2\, d\abs{\mu} \\
		& \leq \sum_{n\in \Z} \norm{u}_{\infty,[n\delta,(n+1)\delta]}^2 \norm{\mu}_{\loc} \\
		& \leq \norm{\mu}_{\loc} \sum_{n\in \Z} \left(4 \delta \norm{u'}_{L_2(n\delta,(n+1)\delta)}^2 + \frac{4}{\delta} \norm{u}_{L_2(n\delta,(n+1)\delta)}^2\right)  \\
		& = 4\delta\norm{\mu}_{\loc} \norm{u'}_2^2 + \frac{4\norm{\mu}_{\loc}}{\delta} \norm{u}_2^2.
	\end{align*}
\end{proof}

Let $\mu$ be a Gordon measure and define
\begin{align*}
	D(\form) & := W_2^1(\R),\\
	\form(u,v) & := \int u'\overline{v}' + \int u\overline{v}\, d\mu.
\end{align*}
Then $\form$ is a closed symmetric semibounded form. Let $H$ be the associated 
self-adjoint operator.

In \cite{BenAmorRemling2005}, Ben Amor and Remling introduced a direct approach
for defining the Schr\"odinger operator $H = -\Delta + \mu$. Since we will
use some of their results we sum up the main ideas: For $u\in W_{1,\loc}^1(\R)$ define $Au\in L_{1,\loc}(\R)$ by
\[Au(x) := u'(x) - \int_{0}^x u(t)\, d\mu(t),\]
where
\[\int_{0}^x u(t)\, d\mu(t) := \begin{cases}
                               	\int_{[0,x]} u(t)\, d\mu(t) & \mbox{if}\; x\geq 0,\\
				-\int_{(x,0)} u(t)\, d\mu(t) & \mbox{if}\; x< 0.
                              \end{cases}\]
Clearly, $Au$ is only defined as an $L_{1,\loc}(\R)$-element. We define the operator $T$ in $L_2(\R)$ by
\begin{align*}
	D(T) & := \set{u\in L_2(\R);\; u,Au\in W_{1,\loc}^1(\R),\, (Au)'\in L_2(\R)},\\
	Tu & := -(Au)'.
\end{align*}

\begin{lemma}
\label{lem:Hsubseteq_T}
	$H\subseteq T$.
\end{lemma}

\begin{proof}
	Let $u\in D(H)$. Then $u\in W_{2}^1(\R)\subseteq W_{1,\loc}^1(\R)$ and 
	$Au\in L_{1,\loc}(\R)$. Let $\varphi\in C_c^\infty(\R)\subseteq D(\form)$. 
	Using Fubini's Theorem, we compute
	\begin{align*}
		& \int_{\R} (Au)(x) \varphi'(x)\, dx \\
		& = \int_{\R} \left(u'(x) - \int_{0}^x u(t)\, d\mu(t)\right)\varphi'(x)\, dx \\
		& = \int_{\R} u'(x)\varphi'(x)\, dx - \int_{\R}\int_{0}^x u(t)\, d\mu(t) \varphi'(x)\, dx \\
		& = \int_{\R} u'(x)\varphi'(x)\, dx
		+ \int_{-\infty}^0\int_{(-\infty,t)}\varphi'(x)\, dx  u(t)\, d\mu(t) - \int_{0}^\infty\int_{[t,\infty)}\varphi'(x)\, dx  u(t)\, d\mu(t)\\
		& = \int_{\R} u'(x)\varphi'(x)\, dx + \int_{-\infty}^0 u(t) \varphi(t)\, d\mu(t) + \int_{0}^\infty u(t)\varphi(t)\, d\mu(t)\\
		& = \int_{\R} u'\varphi' + \int_{\R} u(x)\varphi(x)\, dx 
		 = \form(u,\overline{\varphi}) = \sp{Hu}{\overline{\varphi}} 
		 = \int_\R Hu(x) \varphi(x)\, dx.
	\end{align*}
	Hence, $(Au)' = -Hu\in L_2(\R)$. We conclude that $Au\in W_{1,\loc}^1(\R)$ and therefore $u\in D(T)$, $Tu = -(Au)' = Hu$.
\end{proof}

\begin{remark}
	For $u\in D(H)$ we obtain
	\[u'(x) = Au(x) + \int_{0}^x u(t)\, d\mu(t)\]
	for a.a. $x\in\R$. Since $Au\in W_{1,\loc}^1(\R)$ and $x\mapsto \int_{0}^x u(t)\, d\mu(t)$ is continuous at all $x\in\R$ with $\mu(\set{x}) = 0$, $u'$ is continuous at $x$ for all $x\in \R\setminus \spt \mu_p$, where $\mu_p$ is the point measure part of $\mu$.
\end{remark}

\section{Absence of eigenvalues}
\label{sec:solution_estimates}

We show that $H$ has no eigenvalues. The proof is based on two observations. 
The first one is a stability result and will be achieved in Lemma 
\ref{lem:expest}, the second one is an estimate of the solution for periodic 
measure perturbations, see Lemma \ref{lem:estimate}.

As in \cite{DamanikStolz2000} we start with a Gronwall Lemma, but in a more general version for locally finite measures. For the proof, see \cite{EthierKurtz2005}. 

\begin{lemma}[Gronwall]
\label{lem:gronwall}
	Let $\mu$ be a locally finite Borel measure on $[0,\infty)$, $u\in \mathcal{L}_{1,\loc}([0,\infty),\mu)$ and $\alpha:[0,\infty)\to [0,\infty)$ measurable. Suppose, that
	\[u(x) \leq \alpha(x) + \int_{[0,x]} u(s)\, d\mu(s) \quad(x\geq0).\]
	Then
	\[u(x) \leq \alpha(x) + 
\int_{[0,x]}\alpha(s)\exp\bigl(\mu([s,x])\bigr)\,d\mu(s) \quad(x\geq0).\]
\end{lemma}

For $x\in\R$ we abbreviate
\[I_x:= [x\wedge0,x\vee0]\]
and
\[I_x(t):= I_x\cap([t,x]\cup[x,t]) \quad(t\in\R).\]

Let $\mu$ be uniformly locally bounded. Then
\[\abs{\mu}(I_x) \leq (\abs{x}+1)\norm{\mu}_\loc \quad(x\in \R).\]
Furthermore, if $\mu$ is periodic and locally bounded, $\mu$ is uniformly locally bounded.

Let $H := -\Delta + \mu$ and $E\in \R$. Then $u\in W_{1,\loc}^1(\R)$ ($=D(A)$) is a \emph{solution} of $Hu=Eu$, if $-(Au)' = Eu$ in the sense of distributions (i.e., $u$ satisfies the eigenvalue equation but without being an $L_2$-function).

\begin{lemma}
\label{lem:cont}
	Let $\mu_1,\mu_2$ be two uniformly locally bounded measures, $E\in \R$ and $u_1$ and $u_2$ solutions of
	\[H_1 u_1 = E u_1,\qquad H_2 u_2 = E u_2\]
	subject to
	\[u_1(0) = u_2(0),\quad u_1'(0\rlim) = u_2'(0\rlim),\quad \abs{u_1(0)}^2 + \abs{u_1'(0\rlim)}^2 = 1.\]
	Then there are $C_0,C\geq 0$ such that for all $x\in\R$
	\begin{align*}
		& \norm{\begin{pmatrix}
	         	u_1(x) \\
			u_1'(x)
	        \end{pmatrix}
		- 
		\begin{pmatrix}
	         	u_2(x) \\
			u_2'(x)
	        \end{pmatrix}}  \\
		& \leq C_0 + \int_{I_x} \abs{u_2(t)}\, d \abs{\mu_1-\mu_2}(t) \\
		& + \!C\!\!\int_{I_x} \!\Big(\!C_0 + \int_{I_t} \!\abs{u_2} d \abs{\mu_1-\mu_2}\Big)e^{C(\lambda + \abs{\mu_1 - E\lambda})(I_x(t))} \, d(\lambda + \abs{\mu_1 - E\lambda})(t).
	\end{align*}
\end{lemma}

\begin{proof}
	Write
	\[u_1(x) - u_2(x) = \int_{0}^x (u_1'(t) - u_2'(t))\, dt\]
	and
	\begin{align*}
		u_1'(x) - u_2'(x) & = u_1'(0\rlim) - u_2'(0\rlim)  - \left(u_1(0)\mu_1(\set{0}) - u_2(0)\mu_2(\set{0})\right) \\
		& + \int_0^x u_1(t)\, d\mu_1(t) - \int_0^x u_2(t)\, d\mu_2(t) - \int_0^x E(u_1(t) - u_2(t))\, dt \\
		& = u_2(0)\left(\mu_2(\set{0}) - \mu_1(\set{0})\right) \\
		& + \int_0^x u_2(t)\, d(\mu_1-\mu_2)(t) + \int_0^x (u_1(t)-u_2(t))\, d(\mu_1-E\lambda)(t).
	\end{align*}
	Hence,
	\begin{align*}
		\begin{pmatrix}
			u_1(x) - u_2(x)\\
			u_1'(x) - u_2'(x)
		\end{pmatrix}
		& = \!\begin{pmatrix} 0 \\ u_2(0)\left(\mu_2(\set{0}) - \mu_1(\set{0})\right)\end{pmatrix} + \int_0^x \begin{pmatrix} 0 \\ u_2(t) \end{pmatrix} d(\mu_1-\mu_2)(t) \\
		& + \int_{0}^x \begin{pmatrix} 0 & 1 \\ 1 & 0 \end{pmatrix}\begin{pmatrix} u_1(t) - u_2(t) \\ u_1'(t) - u_2'(t)\end{pmatrix} \, d\begin{pmatrix} \lambda \\ \mu_1 - E\lambda \end{pmatrix}(t).
	\end{align*}
	We conclude, that
	\begin{align*}
		& \norm{\begin{pmatrix}
	         	u_1(x) \\
			u_1'(x)
	        \end{pmatrix}
		- 
		\begin{pmatrix}
	         	u_2(x) \\
			u_2'(x)
	        \end{pmatrix}} \\
		& \leq C_0 + \int_{I_x} \abs{u_2(t)}\, d \abs{\mu_1-\mu_2}(t) \\
		& + C\int_{I_x} \norm{\begin{pmatrix}
	         	u_1(t) \\
			u_1'(t)
	        \end{pmatrix}
		- 
		\begin{pmatrix}
	         	u_2(t) \\
			u_2'(t)
	        \end{pmatrix}}  \, d (\lambda + \abs{\mu_1 - E\lambda})(t).
	\end{align*}	
	An application of Lemma \ref{lem:gronwall} with $\alpha(x) = C_0 + \int_{I_x} \abs{u_2(t)}\, d \abs{\mu_1-\mu_2}(t)$ and $\mu = C(\lambda + \abs{\mu_1 - E\lambda})$ yields the assertion.
\end{proof}

\begin{remark}
\label{rem:further_est}
	Regarding the proof of Lemma \ref{lem:cont} we can further estimate $C_0 \leq \abs{u_2(0)}\abs{\mu_1-\mu_2}(I_x)$ ($x\in\R$).
\end{remark}

\begin{lemma}
	Let $E\in \R$ and $u_0$ be a solution of $-\Delta u_0 = E u_0$. Then 
	there is $C\geq 0$ such that $\abs{u_0(x)} \leq Ce^{C\abs{x}}$ for all 
	$x\in \R$.
\end{lemma}

In the following lemmas and proofs the constant $C$ may change from line to line, but we will always state the dependence on the important quantities.

\begin{lemma}
\label{lem:periodic_solution}
	Let $\mu_1$ be a locally bounded $p$-periodic measure, $E\in\R$, $u_1$ a solution of $H_1u_1 = E u_1$. Then there is 
	$C\geq0$ such that
	\[\abs{u_1(x)} \leq Ce^{C\abs{x}} \quad(x\in\R).\]
\end{lemma}

\begin{proof}
	Let $u_0$ be a solution of $-\Delta u_0 = E u_0$ subject to the same boundary conditions at $0$ as $u_1$.
	By Lemma \ref{lem:cont} we have
	\begin{align*}
		& \abs{u_1(x)-u_0(x)} \\
		& \leq C + \int_{I_x} \abs{u_0(t)}\, d \abs{\mu_1}(t) \\
		& + C\int_{I_x} \left(C + \int_{I_t} \abs{u_0(s)}\, d \abs{\mu_1}(s)\right)e^{C(\lambda + \abs{\mu_1 - E\lambda})(I_x(t))}\, d(\lambda + \abs{\mu_1 - E\lambda})(t) \\
		& \leq  C + \abs{\mu_1}({I_x})Ce^{C\abs{x}} \\
		& + \int_{I_x} \left(C + C\abs{\mu_1}({I_t})e^{C\abs{t}}\right) e^{(\lambda + \abs{\mu_1 - E\lambda})(I_x(t))}\, d(\lambda + \abs{\mu_1 - E\lambda})(t) \\
		& \leq \left(C + C\abs{\mu_1}({I_x})e^{C\abs{x}}\right) \left(1+e^{(\lambda + \abs{\mu_1 - E\lambda})({I_x})}(\lambda + \abs{\mu_1 - E\lambda})({I_x})\right).
	\end{align*}
	Since $\mu_1$ is periodic and locally bounded it is uniformly locally bounded and we have 
	$\abs{\mu_1}({I_x}) \leq (\abs{x}+1)\norm{\mu_1}_{\loc}$. Furthermore, also
	$\mu_1-E\lambda$ is periodic and uniformly locally bounded, so $\abs{\mu_1
	- E\lambda})({I_x}) \leq (\abs{x}+1)\norm{\mu_1-E\lambda}_{\loc}$. We 
	conclude that
	\begin{align*}
		& \abs{u_1(x)-u_0(x)} \\
		& \leq \left(C + C(\abs{x}+1)\norm{\mu_1}_{\loc}e^{C\abs{x}}\right)\times \\
		& \quad \quad \quad\times \left(1+e^{(\abs{x}+1)(1+\norm{\mu_1-E\lambda}_{\loc})}(\abs{x}+1)(1+\norm{\mu_1-E\lambda}_{\loc})\right) \\
		& \leq C e^{C\abs{x}},
	\end{align*}
	where $C$ is depending on $E$, $\norm{\mu_1}_{\loc}$ and $\norm{\mu_1-E\lambda}_{\loc}$.
	Hence,
	\[\abs{u_1(x)} \leq \abs{u_1(x)-u_0(x)} + \abs{u_0(x)} \leq Ce^{C\abs{x}}.\]
\end{proof}

\begin{lemma}
\label{lem:expest}
	Let $\mu$ be a Gordon measure and $(\mu^{m})$ the $p_m$-periodic 
	approximants, $E\in\R$. Let $u$ be a solution of $Hu = E u$, $u_m$ a 
	solution of $H_mu_m = E u_m$ for $m\in\N$ (obeying the same boundary 
	conditions at $0$).
	Then there is $C\geq 0$ such that
	\[\norm{\begin{pmatrix}
	         	u(x) \\
			u'(x)
	        \end{pmatrix}
		- 
		\begin{pmatrix}
	         	u_m(x) \\
			u_m'(x)
	        \end{pmatrix}} \leq Ce^{C\abs{x}} \abs{\mu - \mu^{m}}(I_x) \quad(x\in\R).\]
\end{lemma}

\begin{proof}
	By Lemma \ref{lem:cont} and Remark \ref{rem:further_est} we know that
	\begin{align*}
		& \norm{\begin{pmatrix}
	         	u(x) \\
			u'(x)
	        \end{pmatrix}
		- 
		\begin{pmatrix}
	         	u_m(x) \\
			u_m'(x)
	        \end{pmatrix}}\\
		& \leq \abs{u_m(0)}\abs{\mu - \mu^m}(I_x) + \int_{I_x} \abs{u_m(t)}\, d \abs{\mu-\mu^{m}}(t) \\
		& + C\!\int_{I_x} \Big(\abs{u_m(0)}\abs{\mu - \mu^m}(I_t) + \int_{I_t} \abs{u_m}d \abs{\mu-\mu^{m}}\Big)\times \\
		& \qquad\qquad\qquad\qquad\qquad \times e^{C(\lambda + \abs{\mu - E\lambda})(I_x(t))}\, d(\lambda + \abs{\mu - E\lambda})(t).
	\end{align*}
	We have
	\[M:= \sup_{m\in \N} \norm{\mu^{m}}_{\loc} < \infty,\]
	since $(\mu^{m})$ approximates $\mu$. Hence, also 
	\[\sup_{m\in \N} \norm{\mu^{m}-E\lambda}_{\loc} < \infty\]
	and Lemma \ref{lem:periodic_solution} yields
	\[\abs{u_m(x)} \leq Ce^{C\abs{x}},\]
	where $C$ can be chosen independently of $m$. Therefore
	\begin{align*}
		& \norm{\begin{pmatrix}
	         	u(x) \\
			u'(x)
	        \end{pmatrix}
		- 
		\begin{pmatrix}
	         	u_m(x) \\
			u_m'(x)
	        \end{pmatrix}} \\
		& \leq \left(Ce^{C\abs{x}}\abs{\mu-\mu^{m}}({I_x}) + Ce^{C\abs{x}}\abs{\mu-\mu^{m}}({I_x})\right)\times \\
		& \qquad\qquad\qquad\qquad \times \left(1+e^{C(\lambda + \abs{\mu - E\lambda})(I_x)}(\lambda + \abs{\mu - E\lambda})({I_x})\right).
	\end{align*}
	Since 
	\[\abs{\mu - E\lambda}({I_x}) \leq (\abs{x}+1)\norm{\mu-E\lambda}_{\loc},\]
	we further estimate
	\[\norm{\begin{pmatrix}
	         	u(x) \\
			u'(x)
	        \end{pmatrix}
		- 
		\begin{pmatrix}
	         	u_m(x) \\
			u_m'(x)
	        \end{pmatrix}} \leq Ce^{C\abs{x}} \abs{\mu-\mu^{m}}({I_x})\]
	where $C$ is depending on $\norm{\mu-E\lambda}_{\loc}$ (and of course on $M$, $\norm{\mu}_\loc$ and $E$).
\end{proof}

Lemma \ref{lem:expest} can be regarded as a stability (or continuity) result: 
if the measures converge in total variation, the corresponding solutions 
converge as well.

Now, we focus on periodic measures and estimate the solutions.
This will then be applied to the periodic approximations of our Gordon measure 
$\mu$.

\begin{remark}
	\begin{enumerate}
		\item 
			Let $f,g$ be two solutions of the equation $Hu = Eu$. Define their Wronskian by $W(f,g)(x):= f(x)g'(x\rlim) - f'(x\rlim)g(x)$. By \cite{BenAmorRemling2005}, Proposition 2.5, $W(f,g)$ is constant.
		\item
			Let $u$ be a solution of the equation $Hu = Eu$. Define the 
			transfer matrix $T_E(x)$ mapping $(u(0),u'(0\rlim))^\top$ to 
			$(u(x),u'(x\rlim))^\top$. Consider now the two solutions $u_N$, $u_D$
			subject to
			\begin{align*}
				\begin{pmatrix} u_N(0)\\u_N'(0\rlim)\end{pmatrix} = \begin{pmatrix}
				1 \\ 0 	\end{pmatrix},\quad \begin{pmatrix} 
				u_D(0)\\u_D'(0\rlim)\end{pmatrix} = \begin{pmatrix}
								0 \\ 1 	\end{pmatrix}.
			\end{align*}
			Then
			\[T_E(x) = \begin{pmatrix}
			           	u_N(x) & u_D(x) \\
					u_N'(x\rlim) & u_D'(x\rlim)
			           \end{pmatrix}.\]
			We obtain $\det T_E(x) = W(u_N,u_D)(x)$ and $\det T_E$ is constant, hence equals $1$ for all $x\in \R$.
	\end{enumerate}
\end{remark}

\begin{lemma}
\label{lem:estimate}
	Let $\mu$ be $p$-periodic and $E\in \R$. Let $u$ be a solution of $Hu = E u$
	subject to
	\[\abs{u(0)}^2 + \abs{u'(0\rlim)}^2 = 1.\]
	Then
	\[\max\left\{\norm{\begin{pmatrix} u(-p) \\u'(-p\rlim)\end{pmatrix}}, 
	\norm{\begin{pmatrix} u(p) \\u'(p\rlim)\end{pmatrix}}, \norm{\begin{pmatrix} 
	u(2p) \\u'(2p\rlim)\end{pmatrix}}\right\} \geq \frac{1}{2}.\]
\end{lemma}

The proof of this lemma is completely analoguous to the proof of \cite[Lemma 2.2]{DamanikStolz2000}.

\begin{lemma}
\label{lem:u'_konvergiert}
 	Let $v\in L_2\cap BV_{\loc}(\R)$ and assume that for all $r>0$ we have
	\[\abs{v(x) - v(x+r)} \to 0 \quad (\abs{x}\to \infty).\]
	Then $\abs{v(x)}\to 0$ as $\abs{x}\to \infty$.
\end{lemma}

\begin{proof}
	Without restriction, we can assume that $v\geq 0$. We prove this lemma by contradiction. Assume that $v(x)\to 0$ does not hold for $x\to \infty$. Then we can find $\delta>0$ and $(q_k)$ in $\R$ with $q_k\to \infty$ such that $v(q_k)\geq \delta$ for all $k\in\N$.
	By square integrability of $v$ we have $\norm{v\1_{[q_k,q_k+1]}}_{2} \to 0$. Therefore, we can find a subsequence $(r_n)$ of $(q_k)$ satisfying
	\[\norm{v\1_{[r_n,r_n+1]}}_{2} \leq 2^{-\frac{3}{2}n} \quad(n\in \N).\]
	Now, Chebyshev's inequality implies
	\[\lambda(\set{x\in[r_n,r_n+1];\; v(x)\ge 2^{-n}}) \leq 2^{2n} \norm{v\1_{[r_n,r_n+1]}}_{2}^2 \leq 2^{-n} \quad(n\in\N).\]
	Denote $A_n:= \set{x\in[r_n,r_n+1];\; v(x)\ge 2^{-n}} - r_n \subseteq [0,1]$. Then $\lambda(A_n) \leq 2^{-n}$ and
	\[\lambda\left(\bigcup_{n\geq 3} A_n\right) \leq \sum_{n\geq 3} 
	\lambda(A_n) \leq 2^{-2} <1.\]
	Hence, $G:=[0,1]\setminus (\bigcup_{n\geq 3} A_n)$ has positive measure. For $r\in G$, $r>0$ it follows
	\[v(r_n+r)\le 2^{-n} \quad (n\geq 3).\]
	Therefore,
	\[\liminf_{n\to\infty} \abs{v(r_n)-v(r_n+r)} \geq \delta>0,\]
	a contradiction.
\end{proof}

\begin{lemma}
\label{lem:konvergenz}
	Let $\mu$ be a Gordon measure, $E\in \R$, $u\in D(H)$ a solution of $Hu=Eu$. Then $u(x)\to 0$ as $x\to \infty$ and $u'(x)\to 0$ as $x\to \infty$.
\end{lemma}

\begin{proof}
	Since $u\in D(H)\subseteq D(\form)\subseteq W_2^1(\R)$ we have $u(x)\to 0$ as $\abs{x}\to \infty$.
	Lemma \ref{lem:Hsubseteq_T} yields $u\in D(T)$ and $-(Au)' = Hu = Eu$. Let $r>0$. Then, for almost all $x\in \R$,
	\begin{align*}
		u'(x+r) - u'(x) & = Au(x+r) - Au(x) + \int_{(x,x+r]} u(t)\, d\mu(t) \\
		& = \int_{x}^{x+r} (Au)'(y)\, dy + \int_{(x,x+r]} u(t)\, d\mu(t).
	\end{align*}
	Hence,
	\begin{align*}
	 	\abs{u'(x+r) - u'(x)} & \leq \abs{E} \int_x^{x+r} \abs{u(y)}\, dy + \int_{(x,x+r]} \abs{u(t)}\, d\abs{\mu}(t) \\
		& \leq \abs{E}r\norm{u}_{\infty,[x,x+r]} + \norm{u}_{\infty,[x,x+r]}\abs{\mu}([x,x+r]) \\
		& \leq \norm{u}_{\infty,[x,x+r]}\left(\abs{E}r +(r+1)\norm{\mu}_{\loc}\right).
	\end{align*}
	By Sobolev's inequality, there is $C\in \R$ (depending on $r$, but $r$ is fixed anyway) such that
	\[\norm{u}_{\infty,[x,x+r]} \leq C \norm{u}_{W_{2}^1(x,x+r)} \to 0 \quad (\abs{x}\to \infty).\]
	Thus,
	\[\abs{u'(x+r) - u'(x)}\to 0 \quad (\abs{x}\to \infty).\]
	An application of Lemma \ref{lem:u'_konvergiert} with $v:=u'$ yields $u'(x)\to 0$ as $\abs{x}\to \infty$.
\end{proof}

Now, we can state the main result of this paper.

\begin{theorem}
\label{thm:Gordon}
	Let $\mu$ be a Gordon measure. Then $H$ has no eigenvalues.
\end{theorem}

\begin{proof}
	Let $(\mu^m)$ be the periodic approximants of $\mu$. Let $E\in \R$ and $u$ be a solution of $Hu = Eu$. Let $(u_m)$ be the 
	solutions for the measures $(\mu^{m})$.
	By Lemma \ref{lem:expest} we find $m_0\in \N$ such that
	\[\norm{\begin{pmatrix}
	         	u(x) \\
			u'(x)
	        \end{pmatrix}
		- 
		\begin{pmatrix}
	         	u_m(x) \\
			u_m'(x)
	        \end{pmatrix}} \leq \frac{1}{4}\]
	for $m\geq m_0$ and almost all $x\in [-p_m,2p_m]$.
	By Lemma \ref{lem:estimate} we have
	\[\limsup_{\abs{x}\to \infty} \left(\abs{u(x)}^2 + \abs{u'(x)}^2\right) 
	\geq \frac{1}{4} > 0.\]
	Hence, $u$ cannot be in $D(H)$ by Lemma \ref{lem:konvergenz}.
\end{proof}

\section{Examples}
\label{sec:examples}

\begin{remark}[periodic measures]
	Every locally bounded periodic measure on $\R$ is a Gordon measure. Thus, for $\mu:= \sum_{n\in\Z} \delta_{n+\frac{1}{2}}$ the operator  $H:= -\Delta + \mu$ has no eigenvalues.
\end{remark}

Some examples of quasi-periodic $L_{1,\loc}$-potentials can be found in 
\cite{DamanikStolz2000}.

For a measure $\mu$ and $x\in\R$ let $T_x\mu:= \mu(\cdot-x)$. If $\mu$ is periodic with period $p$, then $T_p\mu = \mu$.

\begin{example}
	Let $\alpha\in(0,1)\setminus \Q$. There is a unique continued fraction expansion
	\[\alpha = \frac{1}{a_1+\frac{1}{a_2+\frac{1}{a_3+\frac{1}{\ldots}}}}\]
	with $a_n\in\N$. For $m\in\N$ we set $\alpha_m = \frac{p_m}{q_m}$, where
	\begin{align*}
		p_0 & = 0 & p_1 & = 1 & p_m & = a_mp_{m-1}+p_{m-2} \\
		q_0 & = 1 & q_1 & = a_1 & q_m & = a_mp_{m-1}+q_{m-2} .
	\end{align*}
	$\alpha$ is called \emph{Liouville number}, if there is $B\geq 0$ such that
	\[\abs{\alpha-\alpha_m} \leq Bm^{-q_m}.\]
	The set of Liouville numbers is a dense $G_\delta$.

	Let $\nu,\tilde{\nu}$ be $1$-periodic measures and assume that there is $\gamma>0$ such that 
	\[\abs{\nu(\cdot-x)-\nu}([0,1])\leq \abs{x}^\gamma \quad(x\in\R).\]
	Define $\mu:=\tilde{\nu} + \nu\circ \alpha$ and $\mu^{m}:=\tilde{\nu} + \nu\circ\alpha_m$ for $m\in\N$. 
	Then $\mu^m$ is $q_m$-periodic and
	\begin{align*}
		\abs{\mu-\mu^m}([-q_m,2q_m]) & = \abs{\nu\circ \alpha - \nu\circ\alpha_m}([-q_m,2q_m]) \\
		& = \abs{\nu\circ\frac{\alpha}{\alpha_m}-\nu}([-p_m,2p_m]) \\
		& \leq \sum_{n=-p_m}^{2p_m-1} \abs{\nu\circ\frac{\alpha}{\alpha_m}-\nu}([n,n+1]).
	\end{align*}
	Now, we have
	\begin{align*}
		\abs{\nu\circ\frac{\alpha}{\alpha_m}-\nu}([n,n+1]) 
		& = T_{-n}\abs{\nu\circ\frac{\alpha}{\alpha_m}-\nu}([0,1]) \\
		& = \abs{T_{-n}(\nu\circ\frac{\alpha}{\alpha_m})-T_{-n}\nu}([0,1]) \\
		& = \abs{T_{-n}(\nu\circ\frac{\alpha}{\alpha_m})-\nu}([0,1]).
	\end{align*}
	With $g_{m,n}(y):= y+\left(\frac{\alpha}{\alpha_m} - 1\right)(y+n)$ and using periodicity of $\nu$ we obtain
	\[T_{-n}(\nu\circ\frac{\alpha}{\alpha_m}) = \nu\circ g_{m,n}.\]
	Hence,
	\[\abs{\nu\circ\frac{\alpha}{\alpha_m}-\nu}([n,n+1]) = \abs{\nu\circ g_{m,n} - \nu}([0,1]).\]
	For $y\in[0,1]$ and $n\in\set{-p_m,\ldots,2p_m-1}$ we have
	\[\abs{\left(\frac{\alpha}{\alpha_m} - 1\right)(y+n)} \leq \abs{\frac{\alpha}{\alpha_m} - 1}\leq 2q_mBm^{-q_m}.\]
	Thus,
	\[\abs{\nu\circ g_{m,n} - \nu}([0,1]) \leq \left(2q_mBm^{-q_m}\right)^\gamma = (2q_mB)^\gamma m^{-q_m\gamma}.\]
	
	We conclude that
	\[\abs{\mu-\mu^m}([-q_m,2q_m]) \leq 3p_m (2q_mB)^\gamma m^{-q_m\gamma}\]
	and therefore for arbitrary $C\geq 0$
	\[e^{Cq_m}\abs{\mu-\mu^m}([-q_m,2q_m])\to 0 \quad(m\to \infty).\]
	Hence $\mu$ is a Gordon potential and $H:= -\Delta+\mu$ does not have any eigenvalues.
\end{example}

\section*{Acknowledgements}

We want to thank Peter Stollmann for advices and hints and especially for providing a proof of Lemma \ref{lem:u'_konvergiert}. Many thanks to the unkown referee for useful comments.

\bigskip

\noindent
Christian Seifert\\
Fakult\"at Mathematik\\
Technische Universit\"at Chemnitz\\
09107 Chemnitz, Germany \\
{\tt christian.seifert@mathematik.tu-chemnitz.de}


\begin{thebibliography}{00}

\bibitem{BenAmorRemling2005}
A.\;Ben Amor and C.\;Remling,
\emph{Direct and Inverse Spectral Theory of One-dimensional Schr\"odinger Operators with Measures}.
Integr.\ equ.\ oper.\ theory {52} (2005) 395--417.

\bibitem{DamanikStolz2000} 
D.\;Damanik and G.\;Stolz,
\emph{A Generalization of Gordon's Theorem and Applications to quasiperiodic Schr\"odinger Operators}.
El.\ J.\ Diff.\ Equ. {55} (2000) 1--8.

\bibitem{EthierKurtz2005}
S.N.\;Ethier and T.G.\;Kurtz,
\emph{Markov processes, characterization and convergence}.
Wiley Series in Probability and Statistics, 2005.

\bibitem{Gordon1976} 
A.\;Gordon,
\emph{On the point spectrum of the one-dimensional Schr\"odinger operator}.
Usp.\ Math.\ Nauk\ {31} (1976) 257--258.

\bibitem{KlassertLenzStollmann2010}
S.\;Klassert, D.\;Lenz and P.\;Stollmann,
\emph{Delone measures of finite local complexity and applications to spectral theory of one-dimensional continuum models of quasicrystals}.
DCDS-A 29(4) (2011) 1553--1571.

\bibitem{Simon1982} 
B.\;Simon,
\emph{Schr\"odinger Semigroups}.
Bull. Amer. Math. Soc. {7}(3) (1982) 447--526.

\bibitem{StollmannVoigt1996} 
P.\;Stollmann and J.\;Voigt,
\emph{Perturbation of Dirichlet forms by measures}.
Pot. An. {5} (1996) 109--138.

\end{thebibliography}
\end{document}